\documentclass[preprint,authoryear,12pt]{elsarticle}

\usepackage{amssymb}
\usepackage{graphicx}
\usepackage{epsfig}
\usepackage{array}
\usepackage{mathrsfs}
\usepackage{amsfonts}
\usepackage{amsmath}
\usepackage{amssymb}
\usepackage{mathrsfs}
\input xy
\xyoption{all}

\usepackage{tipa}
\usepackage{pifont}
\usepackage{cases}
\usepackage{txfonts}
\usepackage{graphicx}
\usepackage{array}
\usepackage{mathrsfs}
\usepackage{amsfonts}
\usepackage{setspace}
\usepackage{subfigure}

\newtheorem{Theorem}{Theorem}
\newtheorem{Proposition}{Proposition}
\newdefinition{Definition}{Definition}

\newtheorem{Remark}{Remark}
\newtheorem{Example}{Example}

\newtheorem{Algorithm}{Algorithm}
\newproof{proof}{Proof}
\newproof{pot}{Proof of Theorem \ref{thm2}}

\journal{Automatica}

\begin{document}

\begin{frontmatter}
\thispagestyle{empty}


 \title{Bisimilarity Enforcing Supervisory Control for Deterministic
Specifications \tnoteref{label1}}

%

\author[Singapore]{Yajuan Sun}\ead{sunyajuan@nus.edu.sg},    
\author[USA]{Hai Lin}\ead{hlin1@nd.edu},               
\author[Singapore]{Ben M. Chen}\ead{bmchen@nus.edu.sg}  

\address[Singapore]{Dept. of Electrical and Computer Engineering, National University of Singapore, Singapore}  
\address[USA]{Dept. of Electrical Engineering, University of
Notre Dame, USA}

\begin{abstract}

This paper investigates the supervisory control of
nondeterministic discrete event systems to enforce bisimilarity
with respect to deterministic specifications. A notion of
synchronous simulation-based controllability is introduced as a
necessary and sufficient condition for the existence of a
bisimilarity enforcing supervisor, and a polynomial algorithm is
developed to verify such a condition. When the existence condition
holds, a supervisor achieving bisimulation equivalence is
constructed. Furthermore, when the existence condition does not
hold, two different methods are provided for synthesizing maximal
permissive sub-specifications.
\end{abstract}

\newpage
\setcounter{page}{1}


\end{frontmatter}


\label{}

\section{INTRODUCTION}
The notion of bisimulation introduced by
\cite{milner1989communication} has been successfully used as a
behavior equivalence in model checking \citep{clarke1997model},
software verification \citep{chaki2004abstraction} and formal
analysis of continuous \citep{tabuada2004bisimilar}, hybrid
\citep{tabuada2004compositional} and discrete event systems
(DESs). What makes bisimulation appealing is its capability in
complexity mitigation and branching behavior preservation,
specially when we deal with large scale distributed and concurrent
systems such as multi-robot cooperative tasking, networked
embedded systems, and traffic management.



Therefore, recent years have seen increasing research activities
in employing bisimulation to DESs. References
\citep{barrett1998bisimulation}, \citep{komenda2005control} and
\citep{sumodel} used bisimulation for the control of deterministic
systems subject to language equivalence.
\cite{madhusudan2002branching} investigated the control for
bisimulation equivalence with respect to a partial specification,
in which the plant is taken to be deterministic and all events are
treated to be controllable. \cite{tabuada2008controller} solved
the controller synthesis problem for bisimulation equivalence in a
wide variety of scenarios including continuous system, hybrid
system and DESs, in which the bisimilarity controller is given as
a morphism in the framework of category theory.
\cite{zhou2006control} investigated the bisimilarity control for
nondeterministic plants and nondeterministic specifications. A
small model theorem was provided to show that a supervisor
enforcing the bisimulation equivalence between the supervised
system and the specification exists if and only if a state
controllable automaton exists over the Cartesian product of the
system and specification state spaces. This small model theorem
was also extended for partial observation in
\citep{zhou2007small}. In both these works, the existence of a
bisimilarity supervisor depends on the existence of a state
controllable automaton, which is hard to calculate in a systematic
way, and the complexity of checking the existence condition is
doubly exponential. To reduce the computational complexity,
\cite{zhoubisimilarity2011} specialized to deterministic
supervisors. The existence condition for a deterministic
bisimilarity supervisor considering nondeterministic plants and
nondeterministic specifications was identified. Moreover, the
synthesis of deterministic supervisors, feasible supspecifications
and infimal subspecifications were developed as well.
\cite{liu2011bisimilarity} introduced a simulation-based framework
upon which the bisimilarity control for nondeterministic plants
and nondeterministic specifications was studied. In particular, a
new scheme based on the simulation relation was proposed for
synchronization which is different from those commonly used
synchronization operators such as parallel composition and product
in the supervisory control literature.






This paper studies the supervisory control of nondeterministic
plants for bisimulation equivalence with respect to deterministic
specifications. Compared to the existing literature, the
contributions of this paper mainly lie on the following aspects.
First, a novel notion of synchronous simulation-based
controllability is introduced as a necessary and sufficient
condition for the existence of a bisimilarity enforcing
supervisor. Although it is equivalent to the conditions in
\citep{zhoubisimilarity2011} specialized to deterministic
specifications, it provides a great insight into what characters
should a deterministic specification possesses for bisimilarity
control. Second, a test algorithm is proposed to verify the
existence condition, which is shown to be polynomial complexity
(less than the complexity of the conditions in
\citep{zhoubisimilarity2011}). When the existence condition holds,
we further present a systematic way to construct bisimilarity
enforcing supervisors. Third, since a given specification does
always guarantee the existence of a bisimilarity enforcing
supervisor, a key question arises is how to find a maximal
permissive specification which enables the synthesis of
bisimilarity enforcing supervisors. To answer this question, we
investigate the calculation of supremal synchronously
simulation-based controllable sub-specifications by using two
different methods. One is based on a recursive algorithm and the
other directly computes such a sub-specification based on
formulas.




The rest of this paper is organized as follows. Section 2 gives
the preliminary and problem formulation. Section 3 presents the
synthesis of bisimilarity enforcing supervisors. Section 4
investigates the test algorithm for the existence of a
bisimilarity enforcing supervisor. Section 5 explores the
calculation of maximal permissive sub-specifications. This paper
concludes with section 6.

\section{Preliminary and Problem Formulation}

\subsection{Preliminary Results}
A DES is modeled as a nondeterministic automaton $G
=(X,\Sigma,x_{0},\alpha, X_{m})$, where $X$ is the set of states,
$\Sigma$ is the set of events, $\alpha : \!X \times \Sigma \!
\rightarrow 2^X$ is the transition function, $x_0$ is the initial
state and $X_m \subseteq X$ is the set of marked states. The event
set $\Sigma$ can be partitioned into $\Sigma$ = $\Sigma_{uc} \cup
\Sigma_{c}$, where $\Sigma_{uc}$ is the set of uncontrollable
events and $\Sigma_{c}$ is the set of controllable events. Let
$\Sigma^{*}$ be the set of all finite strings over $\Sigma$
including the empty string $\epsilon$. The transition function
$\alpha$ can be extended from events to traces, $\alpha : \!X
\times \Sigma^{*} \!\rightarrow 2^{X}$, which is defined
inductively as: for any $x \in X$, $\alpha(x, \epsilon)=x$; for
any $s\in \Sigma^{*}$ and $\sigma \in \Sigma$, $\alpha(x,
s\sigma)=\alpha(\alpha(x, s), \sigma)$. If the transition function
is a partial map $\alpha: \!X \times \Sigma \!\rightarrow X$, $G$
is said to be a deterministic automaton. For $X_1 \subseteq X$,
the notation $\alpha|_{X_1 \! \times \! \Sigma}$ means $\alpha$ is
restricted from a smaller domain $X_1 \! \times \!\Sigma$ to
$2^{X_1}$. Given $X_1 \subseteq X$, the subautomaton of $G$ with
respect to $X_1$, denoted by $F_{G}(X_1)$, is defined as: $
F_{G}(X_1)= (X_1,\Sigma,x_{0},\alpha_1,X_{m1})$, where $\alpha_1
\! = \!\alpha \!\mid_{X_1 \!\times \! \Sigma}$ and $X_{m1}$ = $X_1
\! \cap X_{m}$. The active event set at state $x$ is defined as
$E_{G}(x)=\{\sigma \in \Sigma~|~\alpha(x, \sigma)$ is defined\}.
Given a string $s \in \Sigma^{*}$, the length of the string $s$,
denoted as $|s|$, is the total numbers of events, and $s(i)$ is
the $i$-$th$ event of this string, where $1 \leq i \leq |s|$.
Given $\Sigma_1 \subseteq \Sigma$, a projection
$P_{\Sigma\!\rightarrow \Sigma_1}$: $\Sigma^{*}\! \rightarrow
\Sigma_{1}^{*}$ is used to filter a string of events from $\Sigma$
to $\Sigma_1$, and it is defined inductively as follows:
$P_{\Sigma\!\rightarrow \Sigma_1}(\epsilon)=\epsilon$; for any
$\sigma \in \Sigma$ and $s \in \Sigma^{*}$,
$P_{\Sigma\!\rightarrow \Sigma_1}(s\sigma)=P_{\Sigma\!\rightarrow
\Sigma_1}(s)\sigma$ if $\sigma \in \Sigma_1$, otherwise,
$P_{\Sigma\!\rightarrow \Sigma_1}(s\sigma)=P_{\Sigma\!\rightarrow
\Sigma_1}(s)$. The language generated by $G$ is defined as
$L(G)=\{s \in \Sigma^{*} \mid \alpha(x_0, s)$ is defined$\}$, and
the marked language generated by $G$ is defined as $L_{m}(G)=\{s
\in \Sigma^{*} \mid \alpha(x_0, s) \cap X_m \neq \emptyset$\}.
Consider three languages $K, K_1, K_2 \subseteq \Sigma^{*}$. The
Kleene closure of $K$, denoted as $K^{*}$, is the language
$K^*=\cup_{n \in \mathbb{N}}K^{n}$, where $K^{0}=\{\epsilon\}$ and
for any $n \geq 0$, $K^{n+1}=K^{n}K$. The prefix closure of $K$,
denoted as $\overline{K}$, is the language $\overline{K}=\{s \in
\Sigma^{*}~|~(\exists t \in \Sigma^{*})~st\in K\}$. The quotient
of $K_1$ with respect to $K_2$, denoted as $K_1/K_2$, is the
language $K_1/K_2=\{s \in \Sigma^{*}~|~(\exists t \in K_2)~st \in
K_1\}$. For two languages $K_1, K_2 \in \Sigma^{*}$ with $K_2
\subseteq K_1 \neq \emptyset$, let $G_{(K_1, K_2)}$ be a
deterministic automaton such that $L(G_{(K_1, K_2)})=K_1$ and
$L_{m}(G_{(K_1, K_2)})=K_2$. For a nondeterministic $G$, let
$det(G)$ be a minimal deterministic automaton such that
$L(det(G))=L(G)$ and $L_{m}(det(G))=L_{m}(G)$.

To model the interaction between automata, we introduce parallel
composition as below \citep{cassandras2008introduction}.

\begin{Definition}\label{parallel}
Given $G_1 =(X_1,\Sigma_1,x_{01},\alpha_1, X_{m1})$ and $G_2
=(X_2,\Sigma_2,x_{02},\alpha_2,X_{m2})$, the parallel composition
of $G_1$ and $G_2$ is an automaton
\[
G_1 || G_2 = ( X_1 \times X_2, \Sigma_1 \cup \Sigma_2,
\alpha_{1||2}, (x_{01}, x_{02}), X_{m1} \times X_{m2}),
\]
where for any $x_1 \in X_1$, $x_2 \in X_2$ and $\sigma \in
\Sigma$, the transition function is defined as:

\[
\alpha_{1||2}((x_1, x_2),\sigma) = \left\{ {\begin{array}{*{20}c}
   \alpha_1(x_1, \sigma) \times \alpha_2(x_2, \sigma) & {\sigma \in E_{G_1}(x_1) \cap E_{G_2}(x_2) };  \\
   \alpha_1(x_1, \sigma) \times \{x_2\} & {\sigma \in E_{G_1}(x_1) \cap \sigma \in E_1 \!\setminus E_2};  \\
   \{x_1\} \times \alpha_2(x_2, \sigma) & {\sigma \in E_{G_2}(x_2) \cap \sigma \in E_2 \!\setminus E_1};  \\
    \emptyset & {otherwise}.  \\
\end{array}} \right.
\]
\end{Definition}

When $\Sigma_1=\Sigma_2$, parallel composition can be understood
as a form of control, where a supervisor is designed to restrict
the behavior of the plant.

Next we present the synchronized state map, which is used to find
the synchronized state pairs of two automata
\citep{zhou2006control}.

\begin{Definition}
Given $G_1 =(X_1,\Sigma_1,x_{01},\alpha_1, X_{m1})$ and $G_2
=(X_2,\Sigma_2,x_{02},\alpha_2,X_{m2})$, the synchronized state
map $X_{synG_1G_2}$: $X_1 \rightarrow 2^{X_2}$ from $G_1$ to $G_2$
is defined as
\[
X_{synG_1G_2}(x_1)=\{x_2 \in X_2~|~(\exists s \in \Sigma^{*})~ x_1
\in \alpha_1(x_{01}, s) \wedge x_2 \in \alpha_2(x_{01}, s)\}.
\]
\end{Definition}


%





Most literature on supervisory control aims to achieve language
equivalence between the supervised system and the specification.
The necessary and sufficient condition for the existence of a
language enforcing supervisor is captured by the notion of
language controllability as below \citep{ramadge1987supervisory}.


\begin{Definition}\label{langc}
Given $G =(X,\Sigma,x_{0},\alpha, X_{m})$, a language $K \subseteq
L(G)$ is said to be language controllable with respect to $L(G)$
and $\Sigma_{uc}$ if
\[
\overline{K}\Sigma_{uc} \cap L(G) \subseteq \overline{K}.
\]
\end{Definition}





As a stronger behavior equivalence than language equivalence,
bisimulation is stated as follows \citep{milner1989communication}.
It is known that bisimulation implies language equivalence and
marked language equivalence, but the converse does not hold.

\begin{Definition}
Given $G_{1} =(X_{1},\Sigma,x_{01},\alpha_{1},X_{m1})$ and $G_{2}
=(X_{2},\Sigma,x_{02},\alpha_{2},X_{m2})$, a simulation relation
$\phi$ is a binary relation $\phi \subseteq X_1 \times X_2$ such
that $(x_{1}, x_{2}) \in \phi$ implies:
\begin{enumerate}
\item[(1)] $(\forall \sigma \in \Sigma)[\forall x_{1}^{'} \in
\alpha_{1}(x_{1},\sigma)\Rightarrow \exists x_{2}^{'} \in
\alpha_{2}(x_{2},\sigma)$ such that $(x_{1}^{'},x_{2}^{'}) \in
\phi]$;

\item[(2)] $x_{1} \in X_{m1} \Rightarrow x_{2} \in X_{m2}$.
\end{enumerate}
\end{Definition}

If there is a simulation relation $\phi$ $\subseteq$ $X_{1} \times
X_{2}$ such that $(x_{01},x_{02}) \in \phi$, $G_{1}$ is said to be
simulated by $G_{2}$, denoted by $G_{1} \prec_{\phi} G_{2}$. For
$\phi \subseteq (X_1 \cup X_2)^{2}$, if $G_{1} \prec_{\phi}
G_{2}$, $G_{2} \prec_{\phi} G_{1}$ and $\phi$ is symmetric, $\phi$
is called a bisimulation relation between $G_{1}$ and $G_{2}$,
denoted by $G_{1} \cong_{\phi} G_{2}$. We sometimes omit the
subscript $\phi$ from $\prec_{\phi}$ or $\cong_{\phi}$ when it is
clear from the context. Then we present a motivating example of
this paper.

\subsection{A Motivating Example}

\begin{figure}[!htb]
\begin{center}
\includegraphics*[scale=.5]{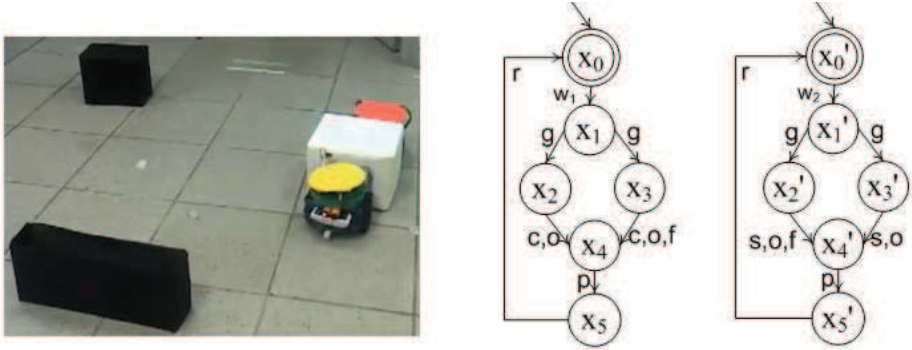}
\caption{ multi-robot system (MRS) (Left), $G_1$ (Middle) and
$G_2$(Right)} \label{plantdet}
\end{center}
\end{figure}
Consider a cooperative multi-robot system (MRS) configured in Fig.
\ref{plantdet} (Left). The MRS consists of two robots $R_1$ and
$R_2$. Both of them have the same communication, position,
pushing, scent-sensing and frequency-sensing capabilities.
Furthermore, $R_1$ has color-sensing capabilities, while $R_2$ has
shape-sensing capability. $R_1$ and $R_2$ can cooperatively search
and clear a dangerous object (the white cube) in the workspace.
Initially, $R_1$ and $R_2$ are positioned outside the workspace.
Let $i=1, 2$. When the work request announces (event $w_i$), $R_i$
is required to enter the workspace. Due to actuator limitations,
it nondeterministically goes along one of two pre-defined paths
(event $g$). In the first path, $R_1$ activates color-sensing
(event $c$) and scent-sensing (event $o$) capabilities to detect
the dangerous object; whereas in the second path, besides
color-sensing and scent-sensing capabilities, $R_1$ also activates
frequency-sensing (event $f$) for detection. Similarly, $R_2$
activates shape-sensing (event $s$), scent-sensing and
frequency-sensing capabilities in the first path, while in the
second path it activates shape-sensing and scent-sensing
capabilities. After detecting the dangerous object, $R_i$ pushes
the dangerous object outward the workspace (event $p$), and then
returns to the initial position (event $r$) for the next
implementation.

\begin{figure}[!htb]
\begin{center}
\includegraphics*[scale=.5]{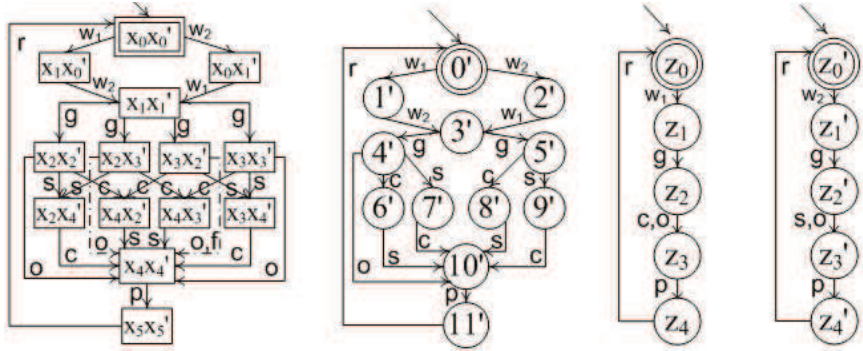}
\caption{$G_1||G_2$ (First Left), $R$ (Second Left), $S_1$ (Second
Right) and $S_2$ (First Right)} \label{spec}
\end{center}
\end{figure}

The automaton model $G_i$ of $R_i$ with alphabet $\Sigma_{i}$ is
shown in Fig. \ref{plantdet}, where $\Sigma_1=\{w_1, g, c, o, f,
p, r\}$ and $\Sigma_2=\{w_2, g, s, o, f, p, r\}$. Since $R_i$ can
not disable the host computer to broadcast the work announcement,
the event $w_i$ is deemed uncontrollable, that is $w_i \in
\Sigma_{uci}$. The rest events are controllable. The cooperative
behavior of $R_1$ and $R_2$ can be represented as $G_1||G_2$ (Fig.
\ref{spec} (First Left)). The specification $R$, configured in
Fig. \ref{spec}, is given in order to restrict the cooperative
behavior $G_1||G_2$. According to the specification, after both
$R_1$ and $R_2$ receive the work command and go to the workspace,
two possible states may be reached by the MRS
nondeterministically. In the first state, the color sensor, the
shape sensor and the scent sensors can be adopted to confirm an
objective is dangerous. However, to save the energy, in the second
state only the color sensor and the shape sensor can be adopted
for dangerous object detection. After the detection, the dangerous
object is cleared from the workspace.



\begin{figure}[!htb]
\begin{center}
\includegraphics*[scale=.5]{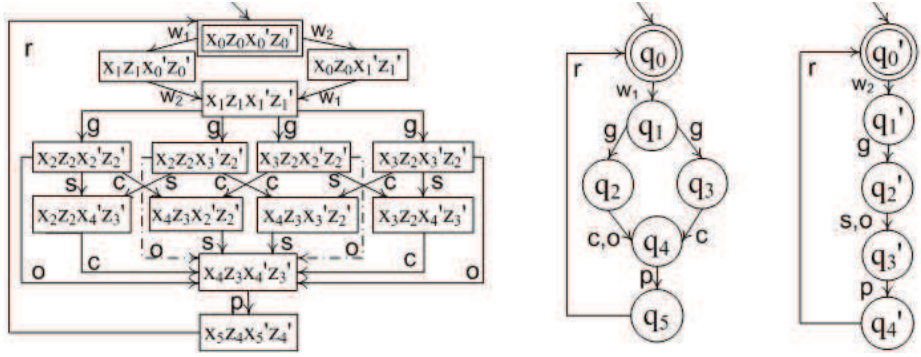}
\caption{$||_{i \in \{1, 2\}} G_i||S_i$ (Left), $R_{s_1}$ (Middle)
and $R_{s_2}$ (Right)} \label{proj}
\end{center}
\end{figure}


For such a MRS, if we use language equivalence as behavior
equivalence, the control target is to design supervisors $S_1$ and
$S_2$ such that $L(\parallel_{i \in \{1, 2\}}G_i||S_i)=L(R)$.
According to the results in \citep{willner1991supervisory}, this
problem can be solved by designing $S_i$ such that
$L(G_i||S_i)=P_{\Sigma_1\!\cup\!\Sigma_2\rightarrow
\Sigma_{i}}(L(R))$. Since $P_{\Sigma_1\!\cup\!\Sigma_2\rightarrow
\Sigma_{i}}(L(R))$ is language controllable with respect to
$L(G_i)$ and $\Sigma_{uci}$, we can construct $S_i$ as shown in
Fig. \ref{spec}. So the supervised system $||_{i \in \{1,
2\}}G_i||S_i$ (Fig. \ref{proj} (Left)) is language equivalent to
$L(R)$. However, it can be seen that $||_{i \in \{1, 2\}}G_i||S_i$
enables all the color sensor, the shape sensor and the scent
sensors for dangerous object detection, which violates the energy
saving requirement in the specification. Hence langauge
equivalence is not adequate for this case, which calls for the use
of bisimulation as behavior equivalence. That is, we need design
supervisor $S_i'$ such that $||_{i \in \{1, 2\}}G_i||S_i' \cong
R$. For such a bisimilarity control problem, a promising method
\citep{karimadini2011guaranteed} is to decompose the global
specification $R$ into sub-specifications $R_{s_i}$ with alphabet
$\Sigma_{i}$ for $R_i$ (Fig. \ref{proj}) such that $||_{i \in \{1,
2\}}R_{s_i}\cong R$ . If we can design $S_i'$ such that $G_i||S_i'
\cong R_{s_i}$, then $||_{i \in \{1, 2\}}G_i||S_i' \cong R$. In
particular, $R_{s_2}$ is deterministic, which motivates us to
consider the bisimilarity control for deterministic specifications
in this paper.


%

\subsection{Problem Formulation}
In the rest of paper, unless otherwise stated we will use $G=(X,
\Sigma, \alpha, x_{0}, X_m)$, $R=(Q, \Sigma, \delta, q_{0}, Q_m)$
and $S=(Y, \Sigma, \beta, y_0, Y_m)$ to denote the
nondeterministic plant, the deterministic specification and the
supervisor (possibly nondeterministic) respectively. Next we
formalize the notion of bisimilarity enforcing supervisor, which
always enables all uncontrollable events and enforces bisimilarity
between the supervised system and the specification.


\begin{Definition}\label{bissup}
Given a plant $G$ and a specification $R$, a supervisor $S$ is
said to be a bisimilarity enforcing supervisor for $G$ and $R$ if:

(1) There is a bisimulation relation $\phi$ such that $G||S
\cong_{\phi} R$;

(2) $(\forall y \in Y)(\forall \sigma \in \Sigma_{uc}) ~\beta(y,
\sigma) \neq \emptyset$.

\end{Definition}


This paper aims to solve the following problems.

{\bf Problem 1:} Given a nondeterministic plant $G$ and a
deterministic specification $R$, what condition guarantees the
existence of a bisimilarity enforcing supervisor $S$ for $G$ and
$R$?

{\bf Problem 2:} How to check this condition effectively?

{\bf Problem 3:} If the condition is satisfied, how to construct a
bisimilarity enforcing supervisor $S$?

{\bf Problem 4:} If the condition is not satisfied, how to obtain
a maximal permissive sub-specification which enables the synthesis
of bisimilarity enforcing supervisors?

\section{ Supervisory Control for Bisimilarity}




%
%
%

This section investigates Problem 1 and Problem 3, also called the
bisimilarity enforcing supervisor synthesis problem. We begin with
the existence condition of a bisimilarity enforcing supervisor.
For sufficiency, since we need design a bisimilarity enforcing
supervisor, the following concept is introduced.

\begin{Definition}
Given $G_1 =(X_1,\Sigma,x_{01},\alpha_1, X_{m1})$, the
uncontrollable augment automaton $G_{1uc}$ of $G_1$ is defined as:
\[
G_{1uc} =(X_1 \cup \{D_d \},\Sigma,x_{01},\alpha_{uc},X_{m1}),
\]
where for any $x \in X_1 \cup \{D_d \}$ and $\sigma \in \Sigma$:
\[
  \alpha_{uc}(x,\sigma) = \left\{ {\begin{array}{*{20}c}
   \alpha_1(x, \sigma) & \sigma \in E_{G_1}(x);  \\
   \{D_d\} & { (\sigma \in \Sigma_{uc} \!\setminus E_{G_1}(x)) \vee (x=D_d \wedge \sigma \in \Sigma_{uc})} ;  \\
   \emptyset & {otherwise.}\\
\end{array}} \right.
\]
\end{Definition}

We can see that an uncontrollable augment automaton can be
employed in the construction of bisimilarity enforcing supervisors
because it naturally satisfies the condition (2) required for a
bisimilarity enforcing supervisor (Definition \ref{bissup}).


On the other side, for necessity we have $G||S \cong R$, which
implies $R \prec G||S \prec G$. Hence $R \prec G$ is a necessary
condition to guarantee the existence of a bisimilarity enforcing
supervisor. Moreover, $G||S \cong R$ implies $L(G||S)=L(R)$, thus
language controllability of the specification is also a necessary
condition for the existence of a bisimilarity enforcing
supervisor. To satisfy those necessary conditions, we will
introduce synchronous simulation-based controllability as a
property of the specification. Before that, we need the following
concept.

\begin{Definition}
Given $G_{1} =(X_{1},\Sigma,x_{01},\alpha_{1},X_{m1})$, $G_{2}
=(X_{2},\Sigma,x_{02},\alpha_{2},X_{m2})$ and a simulation
relation $\phi$ such that $G_{1} \prec_{\phi} G_{2}$, $\phi$ is
called a synchronous simulation relation from $G_1$ to $G_2$ if
$(x_1, x_2) \in \phi$ for any $x_1 \in X_1$ and $x_2 \in
X_{synG_1G_2}(x_1)$.
\end{Definition}

If there exists a synchronous simulation relation $\phi$ from
$G_1$ to $G_2$, $G_1$ is said to be synchronously simulated by
$G_2$, denoted as $G_1 \prec_{syn\phi} G_{2}$. For a deterministic
specification $R$, if $R$ is synchronously simulated by $G$, then
$G$ possesses the branches which are bisimilar to $R$ and the
branches which are outside $L(R)$. Hence it turns out that $G||R
\cong R$. If $R$ is further language controllable with respect to
$L(G)$ and $\Sigma_{uc}$, then $G||R=G||R_{uc}$, implying that
$R_{uc}$ is a candidate of bisimilarity enforcing supervisor. Base
on this observation, we provide the following concept.



\begin{Definition}\label{syncdef}
Given $G_1 =(X_1,\Sigma,x_{01},\alpha_1,X_{m1})$ and $G_2
=(X_2,\Sigma,x_{02},\alpha_2,X_{m2})$, $G_1$ is said to be
synchronously simulation-based controllable with respect to $G_2$
and $\Sigma_{uc}$ if it satisfies:

(1) There is a synchronous simulation relation $\phi$ such that
$G_1 \prec_{syn\phi} G_{2}$;

(2) $L(G_1)$ is language controllable with respect to $L(G_2)$ and
$\Sigma_{uc}$.

\end{Definition}

It is immediate to see that when $R$ is synchronously
simulation-based controllable with respect to $G$ and
$\Sigma_{uc}$, it not only satisfies the necessary conditions ($R
\prec G$ and language controllability of $L(R)$) for the existence
of a bisimilarity enforcing supervisor but also enables the
development of $R_{uc}$ as a bisimilarity enforcing supervisor to
accomplish the sufficiency of the existence condition.

Then we present a necessary and sufficient condition for the
existence of a bisimilarity enforcing supervisor.

\begin{Theorem}\label{t}
Given a plant $G$ and a deterministic specification $R$, there
exists a bisimilarity enforcing supervisor $S$ for $G$ and $R$ if
and only if $R$ is synchronously simulation-based controllable
with respect to $G$ and $\Sigma_{uc}$.
\end{Theorem}

\begin{proof}
For sufficiency, we choose $R_{uc}$ as the supervisor. Let
$G||R=(X_{||}, \Sigma, (x_0, q_0),$ $ \alpha_{||},X_{m||})$.
Consider a relation $\phi_1=\{((x, q), q) ~|~ (x, q) \in
X_{||}\}$. We show that $\phi_1 \cup \phi_1^{-1}$ is a
bisimulation relation from $G||R$ to $R$. First note that $((x_0,
q_0), q_0) \in \phi_1$. Pick $((x, q), q) \in \phi_1$ and $(x',
q') \in \alpha_{||}((x, q), \sigma)$, where $\sigma \in \Sigma$.
By the definition of parallel composition, we have $q' \in
\delta(q, \sigma)$, which implies $((x', q'), q') \in \phi_1$.
When $(x', q') \in X_{m||}$, then $q' \in Q_{m}$. On the other
side, pick $(q, (x, q)) \in \phi_1^{-1}$ and $q' \in \delta(q,
\sigma)$. Since $(x, q) \in X_{||}$ and there is a synchronous
simulation relation $\phi$ such that $R \prec_{syn\phi} G$, we
have $(q, x) \in \phi$. Then there is $x' \in \alpha(x, \sigma)$
such that $(q', x') \in \phi$, and if $q' \in Q_{m}$, then $x' \in
X_{m}$. It follows that $(x', q') \in \alpha_{||}((x, q), \sigma)$
and $(x', q') \in X_{m||}$ when $q' \in Q_{m}$. That is, $(q',
(x', q')) \in \phi_{1}^{-1}$. Hence $G || R \cong_{\phi_1 \cup
\phi_1^{-1}} R$. Moreover from determinism and language
controllability of $R$ and the fact that $R_{uc}$ adds every state
a transition to $D_d$ through undefined uncontrollable events does
not change the result of parallel composition, we have
$G||R_{uc}=G||R$. It implies that $G||R_{uc} \cong_{\phi_1 \cup
\phi_1^{-1}}R$.

For necessity, suppose there is a bisimilarity enforcing
supervisor $S$ for $G$ and $R$. Then, there is a bisimulation
relation $\phi'=\phi \cup \phi^{-1}$ such that $R \prec_{\phi}
G||S$ and $G||S \prec_{\phi^{-1}} R$. Let $G||S=(X_{G||S}, \Sigma,
(x_0, y_0), \alpha_{G||S},X_{mG||S})$. Consider a relation
$\phi_1=\{(q, x) \in Q \times X ~|~(\exists y \in Y)~(q, (x, y))
\in \phi \}$. We show that $\phi_1$ is a synchronous simulation
relation from $R$ to $G$. By the definition of parallel
composition, $\phi_1$ is a simulation relation from $R$ to $G$.
Assume there is $q \in Q$ and $x' \in X_{synRG}(q)$ such that $(q,
x') \notin \phi_1$. Hence there exists $s \in \Sigma^{*}$ such
that $q \in \delta(q_0, s)$ and $x' \in \alpha(x_0, s)$. Since $R
\prec_{\phi} G||S$, for $q \in \delta(q_0, s)$, there is $(x, y)
\in \alpha_{G||S}((x_0, y_0), s)$ such that $(q, (x, y)) \in
\phi$, which implies $y \in \beta(y_0, s)$ and in turn implies
$(x', y) \in \alpha_{G||S}((x_0, y_0),s)$. Because $G||S
\prec_{\phi^{-1}} R$, for $(x', y) \in \alpha_{G||S}((x_0,
y_0),s)$, there is $q' \in \delta(q_0, s)$ such that $((x', y),
q') \in \phi^{-1}$. Since $R$ is deterministic, we have $q=q'$.
Therefore, $(q, (x', y)) \in \phi$, which implies $(q, x') \in
\phi_1$. It introduces a contradiction. Then the assumption is not
correct. That is, for any $q \in Q$ and $x \in X_{synRG}(q)$, $(q,
x) \in \phi_1$. So $R \prec_{syn\phi_1} G$. Next we show language
controllability of $L(R)$. Since a bisimilarity enforcing
supervisor $S$ enables all uncontrollable events at each state,
$L(G||S)$ is language controllable with respect to $L(G)$ and
$\Sigma_{uc}$, further, $G||S \cong R$ implies $L(G||S)=L(R)$. It
follows that $L(R)$ is language controllable w.r.t. $L(G)$ and
$\Sigma_{uc}$. So $R$ is synchronously simulation-based
controllable w.r.t. $G$ and $\Sigma_{uc}$.
\end{proof}

\begin{Remark}
Theorem \ref{t} shows that if a deterministic $R$ is synchronously
simulation-based controllable with respect to $G$ and
$\Sigma_{uc}$, $R_{uc}$ is a bisimilarity enforcing supervisor for
$G$ and $R$. Here synchronous simulation-based controllability of
$R$ is equivalent to the conditions ($G||det(R)\cong R$ and
language controllability of $L(R)$) specialized to deterministic
specifications \citep{zhoubisimilarity2011} to ensure the
existence of a deterministic bisimilarity supervisor. However, the
notion of synchronous simulation-based controllability offers
computation advantages compared to the conditions in
\citep{zhoubisimilarity2011} (See section 4). Moreover, it enables
the calculation of maximal permissive sub-specification when the
existence condition for a bisimilarity enforcing supervisor does
not hold (See section 5).
\end{Remark}

\begin{figure}[!htb]
\begin{center}
\includegraphics*[scale=.5]{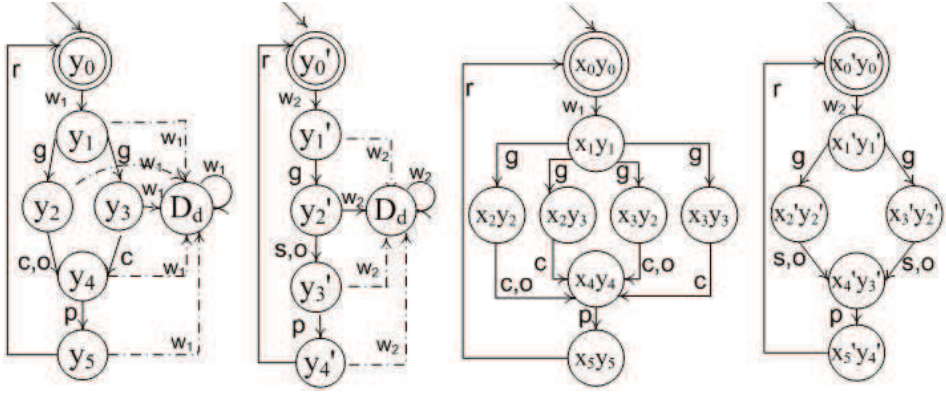}
\caption{ $S_1'$ (First Left), $S_2'$ (Second Left), $G_1||S_1'$
(Second Right) and $G_2||S_2'$ (First Right)} \label{moticls}
\end{center}
\end{figure}

Now we revisit the motivating example.

\begin{Example}\label{lcbis}
Let $i\!=1,2$. We need design supervisor $S_{i}'$ such that
$G_i||S_i' \cong R_{s_i}$. Since $R_{s_2}$ is deterministic and
synchronously simulation-based controllable with respect to $G_2$
and $\Sigma_{uc2}\!=\{w_2\}$, from Theorem \ref{t} we can design
$(R_{s_2})_{uc}$ to be $S_2'$ (Fig. \ref{moticls} (Second Left)).
The supervised system $G_2||S_2'$ is shown in Fig. \ref{moticls}
(First Right) and it can be seen that $G_2||S_2'\! \cong_{\phi
\cup \phi^{-1}} R_{s_2}$, where $\phi\!=\{(q_0', (x_0', y_0')),
(q_1', (x_1', y_1')),(q_2', (x_2', y_2')), (q_2', $ $(x_3',
y_2')),(q_3', (x_4', y_3')), (q_4', (x_5', y_4'))\}$. In addition,
$S_1'$ for $G_1$ can be designed as shown in Fig. \ref{moticls}
(First Left) according to our results in
\citep{sun2012bisimilarityacc}. Then $G_1||S_1' \cong R_{s_1}$
(Fig. \ref{moticls} (Second Right)). As a result, $||_{i \in \{1,
2\}}G_i||S_i' \cong R$.
\end{Example}

\section{A Test Algorithm for the Existence of a Bisimilarity
Enforcing Supervisor}

To solve Problem 2, an algorithm is proposed in this section to
test the existence of a bisimilarity enforcing supervisor. We
start by introducing synchronously simulation-based controllable
product, which will be used in the test algorithm.

\begin{Definition}
Given $G_1 =(X_1,\Sigma,x_{01},\alpha_1, X_{m1})$ and $G_2
=(X_2,\Sigma,x_{02},\alpha_2,X_{m2})$, the synchronously
simulation-based controllable product of $G_1$ and $G_2$ is an
automaton
\[
G_1 ||_{sync} G_2 = ( (X_1 \times X_2) \cup \{q_d, q_d'\}, \Sigma,
\alpha_{12}, (x_{01}, x_{02}), X_{m1} \times X_{m2}),
\]
where for any $(x_1, x_2) \in X_1 \times X_2$ and $\sigma \in
\Sigma$, the transition function is defined as:

\[
\alpha_{12}\!((x_1\!, \!x_2),\!\sigma) = \left\{
{\begin{array}{*{20}c}
   \!\alpha_1(\!x_1, \!\sigma) \! \times\! \alpha_2\!(x_2,\! \sigma) \! & \! {\sigma \in E_{G_1}(x_1)\cap E_{G_2}(x_2) };  \\
   q_d & {\sigma \! \in \! E_{G_1}(x_1)\! \setminus \! E_{G_2}(x_2)};  \\
   q_d' & {\sigma \! \in \! \Sigma_{uc}\cap(E_{G_2}(x_2)\! \setminus E_{G_1}(x_1))}; \\
   \emptyset & {otherwise}. \\
\end{array}} \right.
\]
\end{Definition}

Since synchronous simulation-based controllability is a necessary
and sufficient condition for the existence of a bisimilarity
enforcing supervisor, the following algorithm for testing
synchronous simulation-based controllability of $R$ also verifies
the existence of a bisimilarity enforcing supervisor for $G$ and
$R$.


\begin{Algorithm}\label{algsync}
Given a plant $G$ and a deterministic specification $R$, the
algorithm for testing synchronous simulation-based controllability
of $R$ with respect to $G$ and $\Sigma_{uc}$ is described as
below.

Step 1: Obtain $ R ||_{sync} G=(X_{sync}, \Sigma, \alpha_{sync},
(q_0, x_0), X_{msync})$;

Step 2: $R$ is synchronously simulated-based controllable with
respect to $G$ and $\Sigma_{uc}$ if and only if $q_d$ and $q_d'$
are not reachable in $R||_{sync} G$ and $x \in X_{m}$ for any
reachable state $(q, x)$ in $R||_{sync} G$ with $q \in Q_{m}$.
\end{Algorithm}


\begin{Theorem}\label{alg1c}
Algorithm 1 is correct.
\end{Theorem}
\begin{proof}
From the definition of synchronously simulation-based controllable
product, it is obvious that any $(q, x)$ satisfying $x \in
X_{synRG}(q)$ is a state reachable in $R||_{sync}G$, and any $(q,
x) \in X_{sync}\!\setminus \!\{q_d, q_d'\}$ satisfies that $x \in
X_{synRG}(q)$. For synchronous simulation-based controllability to
hold, condition (1) and condition (2) of Definition \ref{syncdef}
should be satisfied. On the other hand, if condition (1) is
violated, there are two cases. Case 1: there exist $(q, x)$ and
$\sigma \in \Sigma$ such that $x \in X_{synRG}(q)$ and $\sigma \in
E_{R}(q) \!\setminus E_{G}(x)$. So $q_d \in \alpha_{sync}((q, x),
\sigma)$. Case 2: there is $(q, x)$ such that $x \in X_{synRG}(q)$
and $x \notin X_{m}$ when $q \in Q_{m}$. If condition (2) is
violated, i.e. there exist $(q, x)$ and $\sigma \in \Sigma_{uc}$
such that $x \in X_{synRG}(q)$ and $\sigma \in E_{G}(x)
\!\setminus E_{R}(q)$. So $q_d' \in \alpha_{sync}((q, x),
\sigma)$. It follows that $q_d$ and $q_d'$ are reachable in
$R||_{sync} G$ or $x \notin X_{m}$ for any reachable state $(q,
x)$ in $R||_{sync} G$ with $q \in Q_{m}$ iff $R$ is not
synchronously simulated-based controllable w.r.t. $G$ and
$\Sigma_{uc}$.
\end{proof}

\begin{Remark}
Algorithm 1 can be terminated because the state sets and the event
sets of $R$ and $G$ are finite. Since $G$ is nondeterministic and
$R$ is deterministic, their numbers of transitions are $O(|X|^{2}
|\Sigma|)$ and $O(|Q||\Sigma|)$ respectively. Then the complexity
of constructing $R||_{sync}G$ is $O(|X|^{2}|Q|^{2} |\Sigma|)$. In
addition, the complexity of checking the reachability of $q_d$ and
$q_d'$ in $R||_{sync}G$ is $O(log(|X||Q|))$
\citep{jones1975space}. So the complexity of Algorithm 1 is
$O(|X|^{2} |Q|^{2} |\Sigma|)$. That is, the algorithm for testing
the existence of a bisimilarity enforcing supervisor has
polynomial complexity. \cite{zhoubisimilarity2011} used the
conditions such as $G||det (R) \cong R$ and $L(R)$ is language
controllable with respect to $L(G)$ and $\Sigma_{uc}$ to guarantee
the existence of a deterministic supervisor that achieves
bisimulation equivalence. The complexity of verifying those
conditions with respect to deterministic specifications is
$O(|X|^{2}|Q|^{2}|\Sigma|^{3}log (|X||Q|^{2}))$ (Remark 2 in
\citep{zhoubisimilarity2011}). Hence, we argue that Algorithm 1 is
more effective.
\end{Remark}


We provide the following example to illustrate the algorithm for
checking synchronous simulation-based controllability.

\begin{figure}[!htb]
\begin{center}
\includegraphics*[scale=.5]{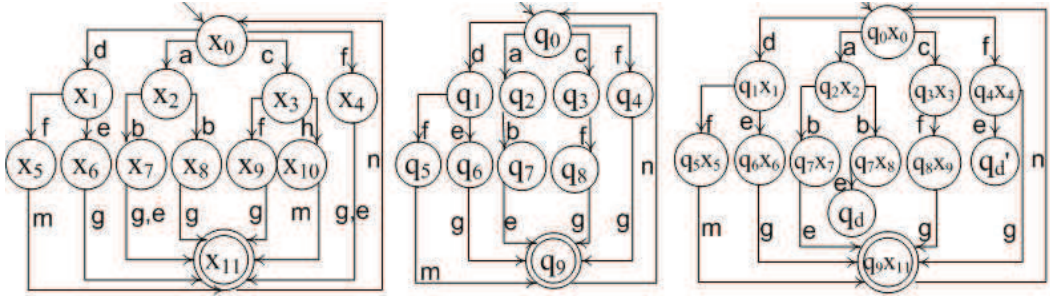}
\caption{ Plant $G$ (Left), Specification $R$ (Middle) and
$R||_{sync} G$ (Right) of Example 2} \label{chk1}
\end{center}
\end{figure}

\begin{Example}\label{checksync}
Consider a plant $G$ and a specification $R$ with
$\Sigma_{uc}=\{b, e\}$ configured in Fig. \ref{chk1}. We can see
that $R$ is not synchronously simulation-based controllable with
respect to $G$ and $\Sigma_{uc}$ because for $f \in L(G)\cap L(R)$
and $e \in \Sigma_{uc}$, $fe \in L(G)\!\setminus L(R)$, and $e$ is
defined at $q_7$ but not $x_8 \in X_{synRG}(q_7)$.

Next we use Algorithm \ref{algsync} to test synchronously
simulation-based controllability of $R$. The synchronously
simulation-based controllable product $R||_{sync}G$ is shown in
Fig. \ref{chk1} (Right). It can be seen that $q_d$ and $q_d'$ are
reachable in $R||_{sync}G$. Hence $R$ is not synchronously
simulation-based controllable with respect to $G$ and
$\Sigma_{uc}$.
\end{Example}

\section{Supremal Synchronously Simulation-Based Controllable Sub-specifications}
This section studies Problem 4, i.e., the synthesis of supremal
synchronously simulation-based controllable sub-specifications,
because a synchronous simulation-based controllable
sub-specification ensures the existence of a bisimilarity
enforcing supervisor. First we introduce the notion of supremal.

Given $(A, \leq)$ and $A' \subseteq A$, where $\leq \subseteq A
\times A$ is a transitive and reflexive relation over $A$, $x \in
A$ is said to be a supremal of $A'$, denoted by $supA'$, if it
satisfies:

(1) $\forall y \in A'$: $y \leq x$;

(2) $\forall z \in A: [\forall y \in A': y \leq z] \Rightarrow [x
\leq z]$.

When we define the supremal of $A'$, a set $(A, \leq)$ should be
given with respect to the element of $A'$. If the elements of $A'$
are languages, the set $(2^{\Sigma^{*}}, \subseteq)$ should be
applied because $2^{\Sigma^{*}}$ includes all languages over
alphabet $\Sigma$ and language inclusion fully captures the
comparison between two languages. However, if the elements of $A'$
are automata, the set $(B, \prec)$ should be applied, where $B$ is
a full set of automata with alphabet $\Sigma$ and $\prec \subseteq
B \times B$ is the simulation relation, since $B$ includes all
automata over alphabet $\Sigma$ and the simulation relation is
adequate for automata (possibly nondeterministic) comparison.


We consider the class of sub-specifications that satisfies
synchronous simulation-based controllability as below.
\begin{eqnarray}
C_1 &:=& \{R' ~|~R'~ is ~deterministic, R' \prec R ~and ~R'~ is~ synchronous~ \nonumber\\
 & &  simulation-based ~controllable~ w.r.t.~ G ~and~
\Sigma_{uc}\} \nonumber
\end{eqnarray}

It can be seen that the supremal of $C_1$ with respect to $(B,
\prec)$ is a supremal synchronously simulation-based controllable
sub-specification. However, it is difficult to directly calculate
the supremal of $C_1$ because $C_1$ is not closed under the upper
bound (join) operator with respect to $(B, \prec)$
\citep{zhoubisimilarity2011}. To encounter this problem, we would
like to convert the automaton set $C_1$ into equivalently
expressed language sets which are closed under the upper bound
(set union) operator with respect to $(2^{\Sigma^{*}}, \subseteq)$
\citep{cassandras2008introduction}. Next we do this conversion
item by item. First, for two deterministic automata $R'$ and $R$,
the condition $R' \prec R$ is equivalent to the language condition
$L(R') \subseteq L(R)$ and $L_{m}(R') \subseteq L_{m}(R)$. Second,
language controllability required in synchronous simulation-based
controllability is naturally a language description. It remains to
convert synchronous simulation relation required in synchronous
simulation-based controllability to an equivalent language
condition. To complete the conversion, we need the following
concept.


\begin{Definition}
Given $G=(X, \Sigma, x_{0}, \alpha, X_{m})$, the synchronous state
merger operator on $G$ is defined as an automaton
\[
F_{syn}(G) = (X_{syn},\Sigma, \{x_{0}\},\alpha_{syn},X_{msyn}),
\]
where $X_{syn} = 2^{X}$, $X_{msyn}=\{ Y_1 ~|~Y_1 \subseteq X_{m}
\}$, and for any $A \in X_{syn}$ and $\sigma \in \Sigma$, the
transition function is defined as:


\[
 \alpha_{syn}(A,\sigma) = \left\{ {\begin{array}{*{20}c}
  \cup_{x \in A} \alpha(x, \sigma) & {\sigma \in \cap_{ x \in A}E_{G}(x) };  \\
   undefined & {otherwise}. \\
\end{array}} \right.
\]
\end{Definition}




By using $F_{syn}(G)$, the synchronous simulation relation from a
deterministic automaton $G_1$ to a plant $G$ is equivalent to
language conditions $L(G_1) \subseteq L(F_{syn}(G))$ and
$L_{m}(G_1) \subseteq L_{m}(F_{syn}(G))$, which is illustrated by
the following proposition.

\begin{Proposition}\label{gsyn}
Given a plant $G$ and a deterministic automaton $G_1$, there is a
synchronous simulation relation $\phi$ such that $G_{1}
\prec_{syn\phi} G$ iff $L(G_1) \subseteq L(F_{syn}(G))$ and
$L_m(G_1) \subseteq L_{m}(F_{syn}(G))$.
\end{Proposition}
\begin{proof}
Let $F_{syn}(G) = (X_{f},\Sigma, \{x_{0}\},\alpha_{f},X_{mf})$,
$G_{1}=(X_1, \Sigma, x_{01}, \alpha_1, X_{m1})$ and
$G_{L}=G_{1}||G=(X_{L},\Sigma, (x_{01}, x_0),\alpha_{L},X_{mL})$.
For sufficiency, consider a relation $\phi=\{(x_1, x) \in X_1
\times X~|~x \in X_{synG_{1}G}(x_1) \}$. We show that $\phi$ is a
synchronous simulation relation from $G_{1}$ to $G$. First note
that $(x_{01}, x_0) \in \phi$. Pick $(x_1, x) \in \phi$ and $x_1'
\in \alpha_{1}(x_1, \sigma)$, where $\sigma \in \Sigma$. Since $x
\in X_{synG_{1}G}(x_1)$, there is $s \in \Sigma^{*}$ such that
$x_1 \in \alpha_1(x_{01}, s)$ and $x \in \alpha(x_0, s)$. Hence
$s, s\sigma \in L(G_1)$, moreover, $L(G_1)\subseteq
L(F_{syn}(G))$. It follows that $s, s\sigma \in L(F_{syn}(G))$.
Therefore there exist $A = \alpha_{f}(\{x_{0}\}, s)$ and $A_1 =
\alpha_{f}(A, \sigma)$. By the definition of $F_{syn}(G)$, we have
$x \in A$ and $\sigma \in \cap_{x''\in A}E_{G}(x'')$, which
implies there is $x' \in \alpha(x, \sigma)$ such that $x' \in
X_{synG_1G}(x_1')$, i.e. $(x_1', x') \in \phi$. Next we show that
$x_1 \in X_{m1}$ implies $x \in X_{m}$. Because $x_1 \in X_{m1}$,
we have $s \in L_m(G_1)$, in addition, $L_m(G_1) \subseteq
L_{m}(F_{syn}(G))$. It follows $s \in L_{m}(F_{syn}(G))$, that is
$A \subseteq X_{m}$, implying $x \in X_{m}$. So
$G_{1}\prec_{syn\phi}G$.

For necessity, the induction method is used to prove $s\in
L(F_{syn}(G))$ for any $s \in L(G_1)$, that is $L(G_1) \subseteq
L(F_{syn}(G))$. (1) $|s|=0$, then $s=\epsilon$. It is obvious that
$\epsilon \in L(F_{syn}(G))$. (2) Assume when $|s|=n$, we have $s
\in L(F_{syn}(G))$ for any $s \in L(G_1)$. (3) $|s|=n+1$. Let
$s=s_1\sigma$, where $\sigma \in \Sigma$. Because $s_1\sigma \in
L(G_1)$ and $G_{1}$ is deterministic, for any $x_2 \in
\alpha_{1}(x_{01}, s_1)$, we have $\sigma \in E_{G_{1}}(x_2)$.
Since $G_{1} \prec_{syn\phi} G$, for any $x'' \in \alpha(x_0,
s_1)$, we have $(x_2, x'') \in \phi$. It follows that $\sigma \in
\cap_{x'' \in \alpha(x_0, s_1)} E_{G}(x'')$. In addition,
$|s_1|=n$ implies $s_1 \in L(F_{syn}(G))$, which in turn implies
there is $A_1 = \alpha_{f}(\{x_0\}, s_1)$ such that $x'' \in A_1$.
Hence $A_2=\alpha_{f}(A_1, \sigma)=\cup_{x''\in A_1} \alpha(x'',
\sigma)$, that is, $s_1\sigma \in L(F_{syn}(G))$. Therefore for
any $s \in L(G_1)$, we have $s \in L(F_{syn}(G))$, i.e. $L(G_1)
\subseteq L(F_{syn}(G))$. Next we show $L_m(G_1) \subseteq
L_{m}(F_{syn}(G))$ by proving $s' \in L_{m}(F_{syn}(G))$ for any
$s' \in L_m(G_1)$. Since $s' \in L_m(G_1)$, there is $x_4
 \in \alpha_1(x_{01}, s')$ such that $x_4 \in X_{m1}$. Because
$G_{1} \prec_{syn\phi} G$ implies $(x_4, x''') \in \phi$ for any
$x''' \in \alpha(x_0, s')$, we have $x''' \in X_{m}$. Definition
of $F_{syn}(G)$ implies $s' \in L_{m}(F_{syn}(G))$, i.e. $L_m(G_1)
\subseteq L_{m}(F_{syn}(G))$.
\end{proof}

Hence the automaton set $C_1$ can be converted into the following
langauge sets:
\begin{eqnarray}
C_2 &:=& \{L_1 \subseteq L(R) \cap L(F_{syn}(G))~|~
L_1=\overline{L_1}
~and~ L_1 ~is ~language ~controllable \nonumber \\
& &  ~w.r.t.~ L(G) ~and~ \Sigma_{uc}\}; \nonumber \\
C_3 &:=& \{L_1 \cap L_{m}(R) \cap L_{m}(F_{syn}(G))~|~L_1 \in
C_2\}. \nonumber
\end{eqnarray}

The computation of supremal synchronously simulation-based
controllable sub-specification, i.e., $supC_1$, with respect to
$(B, \prec)$, can be achieved through the computation of the
supremal languages of $C_2$ and $C_3$ with respect to
$(2^{\Sigma^{*}}, \subseteq)$ as shown in the following theorem.



\begin{Theorem}\label{supeq}
Given a plant $G$ and a deterministic specification $R$, if
$supC_2 \neq \emptyset$, then $G_{(supC_2,\! supC_3)}$ $\in
supC_1$.
\end{Theorem}
\begin{proof}
Let $L_1\!=supC_2\!\neq \emptyset$ and $L_1'\!=supC_2 \!\cap
L_{m}(R) \!\cap L_{m}(F_{syn}(G))\!=supC_3$. First we show that
$G_{(L_1, L_1')}\!\in C_1$. Since $L_1 \!= supC_2$, we have $L_1
\!\in C_2$, which implies $L_1$ is language controllable w.r.t.
$L(G)$ and $\Sigma_{uc}$ and $L_1  \!\subseteq L(F_{syn}(G))$. In
addition, definition of $L_1'$ implies $L_1' \!\subseteq
L_{m}(F_{syn}(G))$. From Proposition \ref{gsyn}, it follows that
$G_{(L_1, L_1')}$ is synchronously simulation-based controllable
w.r.t. $G$ and $\Sigma_{uc}$. Since $L_1 \! \in C_2$ also implies
$L_1 \!\subseteq L(R)$ and $L_1' \!\subseteq L_{m}(R)$ and $R$ and
$G_{(L_1, L_1')}$ are deterministic, we have $G_{(L_1, L_1')}\!
\prec R$. Therefore, $G_{(L_1, L_1')}\! \in C_1$. Next we show
that $R_1 \!\prec G_{(L_1, L_1')}$ for any $R_1 \in C_1$. Suppose
there is $R_1 \! \in C_1$ such that $R_1 \! \nprec G_{(L_1,
L_1')}$. Since $R_1 \! \in C_1$, it implies $R_1 \! \prec R$,
moreover, $R_1$ and $R$ are deterministic. It follows that $L(R_1)
\! \subseteq L(R)$ and $L_{m}(R_1) \! \subseteq L_{m}(R)$. In
addition, $R_1 \! \in C_1$ also implies synchronous
simulation-based controllability of $R_1$. Hence $L(R_1)$ is
language controllable with respect to $L(G)$ and $\Sigma_{uc}$ and
there is a synchronous simulation relation $\phi$ such that $R_1
\! \prec_{syn\phi} G$ implying $L(R_1) \! \subseteq L(F_{syn}(G))$
and $L_{m}(R_1) \!\subseteq L_{m}(F_{syn}(G))$ according to
Proposition \ref{gsyn}. Hence $L(R_1)\! \in C_2$. Moreover,
$L_{m}(R_1) \! \subseteq L(R_1)$. By the definition of supremal,
we have $L(R_1) \! \subseteq supC_2\!= L_1$ and $L_{m}(R_1) \!
\subseteq supC_3 \!=L_1'$, further, $R_1$ and $G_{(L_1, L_1')}$
are deterministic. It follows that $R_1\!\prec G_{(L_1, L_1')}$,
which introduces a contradiction. Hence, the assumption is not
correct. That is, we have $R_1 \! \prec G_{(L_1, L_1')}$ for any
$R_1 \!\in C_1$. So $G_{(L_1, L_1')} \!=G_{(supC_2, supC_3)} \!
\in supC_1$.

\end{proof}


%
%
Next we present a recursive algorithm for computing the supremal
synchronously simulation-based controllable sub-specification.

\begin{Algorithm}\label{alg2}
Given a plant $G$ and a deterministic specification $R$, the
algorithm for computing the supremal synchronously
simulation-based controllable sub-specification with respect to
$G$ and $\Sigma_{uc}$ is described as follows:

Step 1: Obtain $det(G)=(X_{det}, \Sigma, x_{0det}, \alpha_{det},
X_{mdet})$, $G'=(F_{syn}(G)||R)_{uc}=(X', \Sigma, x_{0}', \alpha',
X_{m}')$ and $G''=G'||$ $det(G)=(X'', \Sigma, x_{0}'', \alpha'',
X_{m}'')$;

Step 2: $Z_0:=\{(x_1', x_2) \in X' \times X_{det}~|~x_1'=D_d\}$;

Step 3: $\forall k \geq 0$, $Z_{k+1}=Z_{k} \cup \{z \in
X''-Z_{k}~|~(\exists \sigma \in \Sigma_{uc}) ~ \alpha''(z, \sigma)
\in Z_{k} \}$;

Step 4: If $Z_{k+1}=Z_{k} \neq Z$, then the subautomaton
$F_{G''}(X''-Z_{k})$ of $G''$ is a supremal synchronously
simulation-based controllable sub-specification with respect to
$G$ and $\Sigma_{uc}$.
\end{Algorithm}



\begin{Theorem}\label{calz}
Algorithm \ref{alg2} is correct.
\end{Theorem}
\begin{proof}
Consider $R''\!=F_{G''}(X''\!-Z_{k})\!=(Q'', \Sigma, q_0'',
\delta'', Q_{m}'')$, where $Z_{k+1}\!=Z_{k}\!\neq Z$ with $k
\!\geq $ $0$. First we show that $L(R'') \!\in C_2$. Definition of
$Z_{k}$ implies $L(R'')$ is language controllable w.r.t. $L(G)$
and $\Sigma_{uc}$, and the fact that $L(det(G))\!=L(G)$ implies
$L(R'') \!\subseteq L(F_{syn}$ $ (G)) \!\cap L(R)$ and $L_{m}(R'')
\!\subseteq \!L_{m}(F_{syn}(G)) \!\cap L_{m}(R)$. It follows that
$L(R'') \!\in C_2$. Next we show that $L_2 \!\subseteq L(R'')$ for
any $L_2 \! \in C_2$. Suppose there is $L_2 \! \in C_2$ such that
$L_2 \!\nsubseteq L(R'')$, that is, there is $s \!\in \Sigma^{*}$
such that $s \!\in L_2 \!\setminus L(R'')$. Since $s \! \notin
L(R'')$, there exists $s_1 \!\in \overline{\{s\}}$ such that
$(x_1', x_1) \!\in Z_{k'}$, where $x_1' \!\in \alpha'(x_0', s_1)$,
$x_1 \!\in \alpha_{det}(x_{0det}, s_1)$ and $k'\!=0, 1, \cdots k$.
Hence there is $s_2 \in \Sigma_{uc}^{*}$ such that $x_2' \!\in
\alpha'(x_1', s_2)$ and $x_2 \!\in \alpha_{det}(x_1, s_2)$ with
$(x_2', x_2) \!\in Z_{0}$, which implies $s_1s_2 \!\in L(G)
\!\backslash L(F_{syn}(G)||R)$. Moreover,
$L(F_{syn}(G)||R)\!=L(F_{syn}(G)) \!\cap L(R)$ and $L_2
\!\subseteq L(F_{syn}(G)) \!\cap L(R)$. It follows that $s_1s_2
\!\notin L_2$. If $s_2 \!= \epsilon$, then $s_1 \!\notin L_2$,
which implies $s \!\notin L_2$. If $s_2 \!\neq \epsilon$, then
$s_1s_2(1)\cdots s_2(|s_2|-1) \!\notin L_2$ because $L_2$ is
language controllable w.r.t. $L(G)$ and $\Sigma_{uc}$, $s_2(|s_2|)
\!\in \Sigma_{uc}$ and $s_1s_2 \!\in L(G)\!\setminus L_2$. It in
turn follows that $s_1s_2(1)\!\cdots s_2(|s_2|-2) \!\notin L_2$,
$s_1s_2(1)\!\cdots s_2(|s_2|-3) \!\notin L_2$, $\cdots$, $s_1
\!\notin L_2$. Hence $s \!\notin L_2$. So there is a
contradiction, which implies the assumption is not correct. Then
$L_2 \!\subseteq L(R'')$ for any $L_2 \!\in C_2$. As a result,
$L(R'')\!=supC_2$. It remains to show that $L_{m}(R'')\!=supC_3$.
By the definition of $R''$ and the fact that $L_{m}(F_{syn}(G))
\!\subseteq L_{m}(G)$, we have $L_{m}(R'')\!=L(R'')\!\cap
L_{m}(F_{syn}(G))\! \cap L_{m}(R)\!=supC_2 \!\cap L_{m}$
$(F_{syn}(G)) \!\cap L_{m}(R)\!=supC_3$. It follows that $R''$ is
a deterministic automaton such that $L(R'')\!=supC_2$ and
$L_{m}(R'')$ $=\!supC_3$. By Theorem \ref{supeq}, we have $R''
\!\in supC_1$.

\end{proof}

\begin{Remark}
Algorithm \ref{alg2} can be terminated because the state set $X''$
is finite. Because the state numbers of $F_{syn}(G)$ and $det(G)$
are both $O(2^{|X|})$. Therefore, the complexity of Algorithm
\ref{alg2} is $O(2^{2|X|}|Q||\Sigma|)$.
\end{Remark}





Furthermore, the supremal synchronously simulation-based
controllable sub-specification can be calculated by formulas
without applying the recursive algorithm.

\begin{Theorem}\label{calsupsub}
Given a plant $G$ and a deterministic specification $R$, if
$M=L(R)\cap L(F_{syn}(G))-[(L(G)-L(R) \cap
L(F_{syn}(G)))/\Sigma_{uc}^{*}]\Sigma^{*}\neq\emptyset$, then
$G_{(M, M')}$ is a supremal synchronously simulation-based
controllable sub-specification with respect to $G$ and
$\Sigma_{uc}$, where $M'=M \cap L_{m}(R) \cap L_{m}(F_{syn}(G))$.
\end{Theorem}
\begin{proof}
According to Theorem 1 and Theorem 2 in
\citep{brandt1990formulas}, we obtain $supC_2=L(R) \cap
L(F_{syn}(G))-[(L(G)-L(R) \cap
L(F_{syn}(G)))/\Sigma_{uc}^{*}]\Sigma^{*}=M$. It follows that
$M'=supC_3$. From Theorem \ref{supeq}, $G_{(M, M')}$ is a supremal
synchronously simulation-based controllable sub-specification
w.r.t. $G$ and $\Sigma_{uc}$.
\end{proof}


Now we revisit Example \ref{checksync}.

\begin{figure}[!htb]
\begin{center}
\includegraphics*[scale=.5]{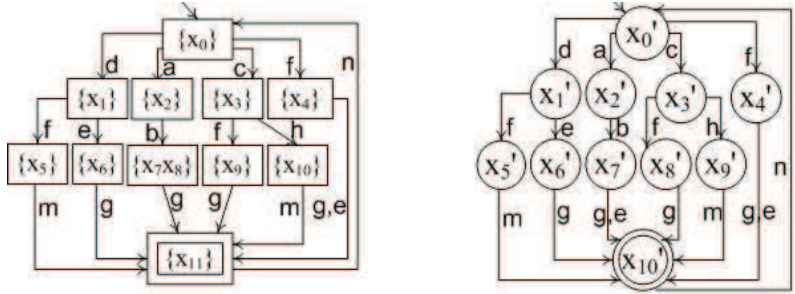}
\caption{ $F_{syn}(G)$ (Left) and $det(G)$ (Right)} \label{chk2}
\end{center}
\end{figure}

\begin{Example}\label{sub2}
Example \ref{checksync} indicates that $R$ is not synchronously
simulation-based controllable with respect to $G$ and
$\Sigma_{uc}$. Thus, we would like to calculate the supremal
synchronously simulation-based controllable sub-specification with
respect to $G$ and $\Sigma_{uc}$ by the proposed methods.

(1) Recursive Method: From Algorithm \ref{alg2}, we establish
$F_{syn}(G)$ and $det(G)$, shown in Fig. \ref{chk2}. Then
$G''\!=(X'', \Sigma, x_{0}'', \alpha'',
X_{m}'')\!=(F_{syn}(G)||R)_{uc}||det(G)$ is achieved in (Fig.
\ref{chksup} (Left)). We obtain $Z_0\!=\!\{(D_{d}, x_{10}')\}$,
$Z_1\!=Z_0 \!\cup \!\{ (\{x_7, x_8\}, q_7, x_7'), (\{x_4\},q_4,
x_4')\}$ \\ and $Z_2\!= \!Z_1 \! \cup \! \{(\{x_2\}, q_2,
x_2')\}\!=\!Z_3$. Therefore, the supremal synchronously
simulation-based controllable sub-specification
$F_{G''}(X''\!-\!Z_{2})$ is obtained in Fig. \ref{chksup}.

\begin{figure}[!htb]
\begin{center}
\includegraphics*[scale=.5]{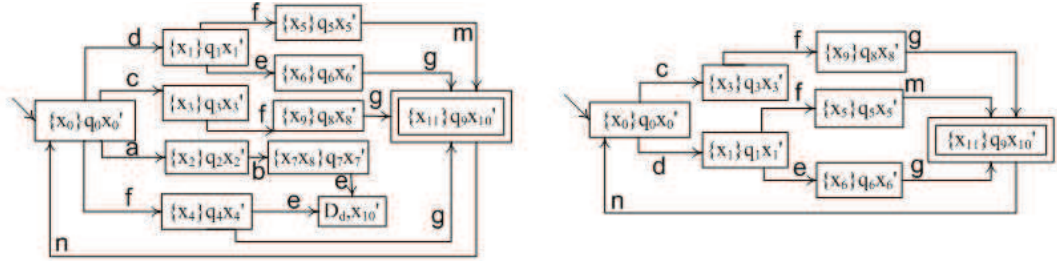}
\caption{ $(F_{syn}(G)|| R)_{uc}||det(G)$ (Left) and
$F_{G''}(X''-Z_2)$ (Right)} \label{chksup}
\end{center}
\end{figure}

(2) Formula-based Method: First we construct $F_{syn}(G)$, which
can be seen in Fig. \ref{chk2} (Left). Hence $L(R) \cap
L(F_{syn}(G)) =\overline{(d(fm+eg)n+cfgn+fgn)^{*}ab}$. Thus,
$M=L(R) \cap L(F_{syn}(G))-[(L(G)-L(R) \cap
L(F_{syn}(G)))/\Sigma_{uc}^{*}]\Sigma^{*}$=$\overline{(d(fm+eg)n+cfgn
}$\\$\overline{+fgn)^{*}ab}$-$(d(fm+eg)n+cfgn+fgn)^{*}ab\Sigma^{*}$
-$(d(fm+eg)n+cfgn+fgn)^{*}a\Sigma^{*}$-$(d(fm+eg)n+cfgn+fgn)^{*}f\Sigma^{*}$
=$\overline{(d(fm+eg)n+cfgn)^{*}}\neq\emptyset$ and $M'=M \cap
L_{m}(R) \cap
L_{m}(F_{syn}(G))$=$(d(fm+eg)n+cfgn)^{*}(d(fm+eg)+cfg)$. The
supremal synchronously simulation-based controllable
sub-specification $G_{(M, M')}\!=F_{G''}(X''-Z_2)$ is achieved in
Fig. \ref{chksup} (Right).
\end{Example}

\section{Conclusion}
In this paper, we investigated the bisimilarity enforcing
supervisory control of nondeterministic plants for deterministic
specifications. A necessary and sufficient condition for the
existence of a bisimilarity enforcing supervisor was deduced from
synchronous simulation-based controllability of the specification,
which can be verified by a polynomial algorithm. For those
specifications fulling the existence condition, a bisimilarity
enforcing supervisor has been constructed. Contrarily, when the
existence condition does not hold, a recursive method and a
formula-based method have been developed to calculate the maximal
permissive sub-specifications.


\bibliographystyle{model5-names}
\bibliography{reference}
\end{document}